\documentclass{article}

\makeatletter{}\usepackage[usenames,svgnames]{xcolor}
\usepackage{graphicx}
\usepackage{tikz}
\usetikzlibrary{decorations.shapes}

\usepackage{wrapfig}
\usepackage{titlesec}
\usepackage{titling}
\usepackage{fullpage}
\usepackage{xspace}
\usepackage[ruled,vlined,linesnumbered]{algorithm2e}
\usepackage{multicol}
\usepackage{framed}
\usepackage{boxedminipage}

\usepackage{amsmath,amssymb,amsthm}
\usepackage[retainorgcmds]{IEEEtrantools}
\usepackage{pdfpages}
\usepackage{lastpage}

\usepackage{ifxetex}

\ifxetex
      \usepackage{fontspec}
      \usepackage{xunicode}
      \defaultfontfeatures{Mapping=tex-text}       \usepackage{euler}
      \usepackage[small,euler-digits]{eulervm}
      \newfontfamily\specialfont{Merge Light}
          \titleformat*{\section}{\Large\bfseries\sffamily\color{SlateGray}\specialfont}
    \titleformat*{\subsection}{\large\bfseries\sffamily\color{DimGray}\specialfont}
    \titleformat*{\subsubsection}{\itshape\specialfont}
    
    \newfontfamily\fancyheading{Lorena}

\else
  \usepackage[utf8]{inputenc}
  \usepackage[T1]{fontenc}
\fi

\newcommand{\shortversion}[1]{}

\usepackage{authblk}

\usepackage{hyperref}

\usepackage{url}

\newtheorem{theorem}{Theorem}
\newtheorem{lemma}{Lemma}

\newtheorem{claim}{Claim}
\newtheorem{corollary}{Corollary}
\newtheorem{definition}{Definition}

\newtheorem{proposition}{Proposition}

\newcommand{\FPT}{\textrm{\textup{FPT}}\xspace}

\newcommand{\NPC}{\textrm{\textup{NP-complete}}\xspace}

\newcommand{\OO}{{\mathcal O}}
\newcommand{\name}[1]{\textsc{#1}}
\newcommand{\YES}{\textsc{Yes}}
\newcommand{\NO}{\textsc{No}}

\newcommand{\problembox}[4]{
\begin{framed}
{\bf #1} \hfill Parameter: {#3} \\
\begin{tabular}{p{.15\textwidth} p{.8\textwidth}}
\hfill Input: & {#2}\\
\hfill Question: & {#4}\\
\end{tabular}
\end{framed}
}

\title{On the Parameterized Complexity of the Maximum Edge Coloring Problem}
\author[$\dag$]{Prachi Goyal}
\author[$\dag$]{Vikram Kamat}
\author[$\dag$]{Neeldhara Misra}
\affil[$\dag$]{Indian Institute of Science, Bangalore}
\affil[ ]{\textit {\{prachi.goyal$|$vkamat$|$neeldhara\}@csa.iisc.ernet.in}}

\date{}

\begin{document}

\maketitle

\begin{abstract}
We investigate the parameterized complexity of the following edge coloring problem motivated by the problem of channel assignment in wireless networks. For an integer $q\geq 2$ and a graph $G$, the goal is to find a coloring of the edges of $G$ with the maximum number of colors such that every vertex of the graph sees at most $q$ colors. This problem is NP-hard for $q \geq 2$, and has been well-studied from the point of view of approximation. Our main focus is the case when $q=2$, which is already theoretically intricate and practically relevant. We show fixed-parameter tractable algorithms for both the standard and the dual parameter, and for the latter problem, the result is based on a linear vertex kernel. 
 \end{abstract}

\makeatletter{}\section{Introduction}

Graph coloring problems are a broad and fundamental class of problems, involving an assignment of colors to the elements of a graph subject to certain constraints. They are often useful in modeling practical questions (map coloring, scheduling, register allocation, and pattern matching, to name a few), and have therefore been of central algorithmic interest. On the other hand, they have also been the subject of intensive structural study.

We are interested in the following edge coloring problem. For an integer $q\geq 2$ and a simple, undirected graph $G=(V,E)$, an assignment of colors to the edges of $G$ is called an edge $q$-coloring if for every vertex $v\in V$, the edges incident on $v$ are colored with at most $q$ colors. An edge $q$-coloring that uses the maximum number of colors is called a maximum edge $q$-coloring. We note that the flavor of this question is quite different from the classical edge coloring question, which is a minimization problem, and the constraints require a vertex to be incident to completely distinct colors. This problem definition is motivated by the problem of channel assignment in wireless networks (as pointed out in~\cite{AP10,FengZW09}, see also~\cite{RaniwalaGC04}). The interference between the frequency channels is understood to be a bottleneck for bandwidth in wireless networks. The goal is to minimize interference to optimize bandwidth. Some wireless LAN standards allow multiple non-overlapping frequency channels to be used simultaneously. In this scenario, a computer on the network equipped with multiple interface cards can use multiple channels. The goal is to maximize the number of channels used simultaneously, if all the nodes in the network have $q$ interface cards. It turns out that the network can be modelled by a simple, undirected graph, while the channel assignment problem corresponds to a coloring of the edges where the edges incident to any given vertex are not colored with more than $q$ colors. The maximum edge $q$-coloring is also considered in combinatorics, as a particular case of the anti-Ramsey number, see~\cite{AP10} for the details of this formulation.

The problem is already interesting for the special case when $q = 2$. It is known to be NP-complete and APX-hard~\cite{AP10} and also admits a $2$-approximation algorithm~\cite{FengZW09}. In their work on this problem, Feng et al~\cite{FengZW09} show the problem to be polynomial time for trees and complete graphs for $q = 2$, and Adamaszek and Popa~\cite{AP10} demonstrate a $5/3$-approximation algorithm for graphs which have a perfect matching. Given these developments, it is natural to pursue the parameterized complexity of the problem. Our main focus will be on the case when $q = 2$, and we note that this special case continues to be relevant in practice. 

The goal of parameterized complexity is to find ways of solving
NP-hard problems more efficiently than brute force. Here the aim is to restrict the combinatorial explosion to a parameter that is hopefully much smaller than the input size. It is a two-dimensional generalization of ``P vs. NP'' where, in addition to the overall input size
$n$, one studies how a secondary measurement (called the
\emph{parameter}), that captures additional relevant information,
affects the computational complexity of the problem in
question. Parameterized decision problems are defined by
specifying the input, the parameter, and the question to be
answered. The two-dimensional analogue of the class~P is
decidability within a time bound of $f(k)n^{c}$, where $n$ is the
total input size, $k$ is the parameter, $f$ is some computable
function and $c$ is a constant that does not depend on $k$ or
$n$. A parameterized problem that can be decided in such a
time-bound is termed \textit{fixed-parameter tractable} (\FPT). For
general background on the theory of fixed-parameter tractability,
see ~\cite{ParameterizedComplexityBook},~\cite{FlumGroheBook}, and~\cite{RN}.

A parameterized  problem is said to admit a {\em polynomial kernel}
if every instance $(I,k)$ can be reduced in polynomial time to an equivalent instance with both size and parameter value bounded by a
polynomial in $k$. The study of kernelization is a major research frontier of  parameterized complexity and many important recent
advances in the area are on kernelization. These include  general results
showing that  certain classes of parameterized problems have polynomial
kernels~\cite{AlonGKSY10,BodlaenderFLPST09,FominLST10}
or randomized kernels~\cite{KratschW12}.
The recent development of a framework for ruling out polynomial kernels under
certain complexity-theoretic
assumptions~\cite{BDF,DM10,FS}   has added a new dimension to
the field and strengthened its connections to classical complexity.  For overviews of kernelization we  refer to surveys~\cite{Bodlaender09,GN07} and to the corresponding chapters in books on parameterized complexity~\cite{FlumGroheBook,RN}.

\paragraph{Our Contributions.} We develop \FPT algorithms and kernels for the maximum edge $2$-coloring problem. The standard parameter is the solution size, or the number of colors used. On the other hand, it is known that the maximum number of colors used in an edge $2$-coloring in a graph on $n$ vertices is at most the number of vertices in the graph. This leads to a natural ``dual'' parameterization below an upper bound. Specifically, we ask if we can color the graph with at least $(n-k)$ colors, and we treat $k$ as the parameter. As an aside, we also characterize the class of graphs that can be colored with $n$ colors as being two-factors (this is implicit in several notions in the literature, and we state the proof for completeness).

Let us consider the problem with the standard parameter. A straightforward and well-understood observation~\cite{AP10,FengZW09} is that the maximum edge $2$-coloring number is at least the size of the maximum matching of the graph. Therefore, if a graph $G$ has a matching of size at least $k$, then $G$ is a YES-instance. This is a simple polynomial time preprocessing step, and therefore we may assume throughout that the size of the maximum matching in the input graph is bounded by $k$. Consequently, the vertex cover of the input is bounded by $2k$ and the treewidth is bounded by $(2k+1)$. We do not consider these natural structural parameterizations separately, since they are implicitly bounded in terms of the solution size.

We note that the expressibility of the maximum edge $2$-colorability question in MSO$_2$ is easily verified. Therefore, we may easily classify the the problem as being \FPT (parameterized by the solution size), by an application of Courcelle's theorem~\cite{Courcelle90}. However, the running time of the algorithm obtained from this meta theorem is impractical, and therefore, we explore the possibility of better algorithms specific to the problem. We first show an exponential kernel obtained by the application of some simple reduction rules, which also implies that the problem is \FPT. We then present a concrete \FPT algorithm that runs in time $\OO^*(k^k)$\footnote{The $\OO^*$ notation is used to suppress polynomial factors in the running time (c.f. Section~\ref{sec:Preliminaries}).} for the problem.
Also, for the dual parameterization, we obtain a linear vertex kernel, with $\OO(k)$ vertices and $\OO(k^2)$ edges. This implies a \FPT algorithm with running time $\OO^*(k^{k^2})$.

This paper is organized as follows. In Section~\ref{sec:Preliminaries} we provide some basic definitions and facts. In section~\ref{sec:mainFPT}, we consider the standard parameter and present an exponential kernel and a \FPT algorithm. In section~\ref{sec:dualparam}, we consider the dual parameter and show a linear vertex kernel. Finally, in section~\ref{sec:kernels} we demonstrate that the problem continues to remain NP-complete on $C_4$-free graphs, and also give a quadratic vertex kernel for this graph class.

\makeatletter{}\section{Preliminaries}\label{sec:Preliminaries}

In this section we state some basic definitions related to
parameterized complexity and graph theory, and give an overview of
the notation used in this paper. To describe running times of
algorithms we sometimes use the $O^{*}$ notation. Given $f:
\mathbb{N} \rightarrow \mathbb{N}$, we define $O^{*}(f(n))$ to be
$O(f(n) \cdot p(n))$, where $p(\cdot)$ is some polynomial
function. That is, the $O^{*}$ notation suppresses polynomial
factors in the running-time expression. The set of natural numbers (that is, nonnegative integers) is denoted by $\mathbb{N}$. For a natural number $n$ let $[n] := \{1, \ldots, n\}$. By $\log n$ we mean $\lceil\log n\rceil$ if an integer is expected.

\paragraph{Graphs.}
In the following, let $G=\left(V,E\right)$ and $G'=\left(V',E'\right)$
be graphs, and $U \subseteq V$ some subset of vertices of $G$. The
union of graphs $G$ and $G'$ is defined as $G\cup G'=\left(V\cup
  V',E\cup E'\right)$, and their intersection is defined as $G\cap
G'=\left(V\cap V',E\cap E'\right)$. A set $U$ is said to be a
\emph{vertex cover} of $G$ if every edge in $G$ is incident to at
least one vertex in $U$. $U$ is said to be an \emph{independent set}
in $G$ if no two elements of $U$ are adjacent to each other. The
\emph{independence number} of $G$ is the number of vertices in a
largest independent set in $G$. $U$ is said to be a \emph{clique} in
$G$ if every pair of elements of $U$ is adjacent to each other. A set
$U$ is said to be a \emph{dominating set} in $G$ if every vertex in
$V\setminus U$ is adjacent to some vertex in $U$. A {\em two-factor} is a graph where every vertex has degree exactly two. We refer the reader to \cite{graphbook} for details on standard graph theoretic notation and terminology we use in the paper.

\paragraph{Parameterized Complexity.} A parameterized problem $\Pi$ is a subset of $\Gamma^{*}\times
\mathbb{N}$, where $\Gamma$ is a finite alphabet. An instance of a
parameterized problem is a tuple $(x,k)$, where $k$ is called the
parameter. A central notion in parameterized complexity is {\em
  fixed-parameter tractability (FPT)} which means, for a given
instance $(x,k)$, decidability in time $f(k)\cdot p(|x|)$, where
$f$ is an arbitrary function of $k$ and $p$ is a polynomial in the
input size. The notion of {\em kernelization} is formally defined
as follows.

\begin{definition}{\rm [\bf
    Kernelization]}~\cite{RN,FlumGroheBook} A kernelization algorithm for a parameterized problem
  $\Pi\subseteq \Gamma^{*}\times \mathbb{N}$ is an algorithm that,
  given $(x,k)\in \Gamma^{*}\times \mathbb{N} $, outputs, in time
  polynomial in $|x|+k$, a pair $(x',k')\in \Gamma^{*}\times
  \mathbb{N}$ such that (a) $(x,k)\in \Pi$ if and only if
  $(x',k')\in \Pi$ and (b) $|x'|,k'\leq g(k)$, where $g$ is some
  computable function.  The output instance $x'$ is called the
  kernel, and the function $g$ is referred to as the size of the
  kernel. If $g(k)=k^{O(1)}$ (resp. $g(k)=O(k)$) then we say that
  $\Pi$ admits a polynomial (resp. linear) kernel.
\end{definition}

\paragraph{The Maximum Edge Coloring Problem.} Let $G = (V,E)$ be a graph, and let $c$ be an assignment of $k$ colors to the edges of $G$, that is, let $c$ be a surjective function from $E$ to $[k]$. We say that $c$ is an edge coloring of the graph using $k$ colors. For a subset $F$ of the edge set $E$, let $c(F)$ denote the set of colors assigned to the edge set $F$, that is,

$$c(F) = \bigcup_{e \in F} c(e).$$

We say that $c$ is {\em $q$-valid} if every vertex in the graph is incident to edges colored with at most $q$ distinct colors. Formally, if $F_v$ denotes the set of edges incident on a vertex $v$, then an edge coloring $c$ is $q$-valid if $|c(F_v)| \leq q$ for all $v \in V$. We denote by $\sigma_q(G)$ the largest integer $k$ for which there exists a $q$-valid edge coloring function with $k$ colors. When considering the special case $q = 2$, we drop the subscript, and simply use $\sigma(G)$ to refer to the maximum number of colors with which $G$ admits a $2$-valid edge coloring. The first algorithmic question that arises is the following:

\problembox{Max Edge $2$-Coloring}{A graph $G$ and an integer $k$}{$k$}{Is $\sigma(G) \geq k$, that is, is there a $2$-valid edge coloring of $G$ with at least $k$ colors?}

We first note that the \name{Max Edge $2$-Coloring} problem is equivalent to its exact version:

\begin{proposition}
For a graph $G$, $\sigma(G) \geq k$ if and only if there is a $2$-valid edge coloring of $G$ with exactly $k$ colors.
\end{proposition}

\begin{proof}
Suppose $\sigma(G) \geq k$, and let $c$ be a $2$-valid edge coloring of $G$ using $(k+r)$ colors, $r \geq 0$. Consider the set of edges $e$ for which $c(e) > k$, that is, consider:

$$F := c^{-1}(k+1) \cup c^{-1}(k+2) \cup \cdots \cup c^{-1}(k+r).$$

Let $c^*$ be the coloring that ``re-colors'' the edges in $F$ with the color $k$, that is:
\begin{equation*}
c^*(e) = \left\{
\begin{array}{cl}
k & \text{if } e \in F\\
c(e) & \text{otherwise.}
\end{array} \right.
\end{equation*}
It is easy to see that $c^*$ is a $2$-valid edge coloring that uses $k$ colors. The other direction of the statement follows directly from the definition of $\sigma(G)$.
\end{proof}

Therefore, when parameterizing by the standard parameter, we will address the question of whether there is a $2$-valid edge coloring that uses exactly $k$ colors, and we refer to this as the \name{Exact Edge $2$-Coloring} problem. We now introduce the dual parameterization. We will need some terminology first. Let $G$ be a graph and let $c: E \rightarrow [k]$ be a $2$-valid edge coloring of $G$ with $k$ colors. For $1 \leq i \leq k$, let $F_i$ denote the set of edges $e$ for which $c(e) = i$, that is, $F_i = c^{-1}(i)$. Notice that each $F_i$ is non-empty. Fix an arbitrary edge $e_i \in F_i$, and let $H$ be the subgraph induced by $\{e_1, \ldots, e_k\}$. We call $H$ the {\em character subgraph} of $G$. Notice that $\Delta(H) \leq 2$. It is also easy to argue that $\sigma(G) \leq |V|$ by examining the character subgraph and using the fact that it has at most $|V|$ edges (see~\cite{FengZW09}). Therefore, we may ask the following question:

\problembox{$(n-k)$-Edge $2$-Coloring}{A graph $G$ and an integer $k$}{$k$}{Is $\sigma(G) \geq (n-k)$, that is, is there a $2$-valid edge coloring of $G$ with at least $(n-k)$ colors?}

An useful notion is that of a {\em palette assignment} associated with an edge coloring $c$. Recall that for a vertex $v$, we use $F_v$ to denote the set of edges incident on $v$. If $c: E \rightarrow [k]$ is an edge coloring, then the palette assignment associated with $c$ is the function $c^{\dag}$ defined as: $c^{\dag}(v) = c(F_v)$. Note that in general, $c^{\dag}$ is a function from $V$ to $2^{[k]}$, however, if $c$ is a $2$-valid coloring, then $c^{\dag}: V \rightarrow {[k] \choose 2} \cup [k] \cup \{\emptyset\}$. We conclude this introduction to the maximum edge coloring problem with a straightforward characterization of graphs for which $\sigma(G) = |V|$.

\begin{proposition}
\label{prop:characterization}
A graph $G=(V,E)$ is a two factor if, and only if, $\sigma(G) = |V|$.
\end{proposition}

\begin{proof}
Let $n := |V|$.

If $G$ is a two-factor, then it has exactly $n$ edges, say $\{e_1, \ldots, e_n\}$. Consider the coloring $c: E \rightarrow [n]$ given by $c(e_i) = i$. Clearly, $c$ is surjective, and since every vertex in $G$ has degree two, any coloring is a $2$-valid coloring. In particular, this shows that $\sigma(G) = n$.

On the other hand, let $\sigma(G) = n$, let $c: E \rightarrow [n]$ be a valid $2$-coloring that uses $n$ colors, and let $H$ be a character subgraph of $G$. Since $|E(H)| \leq n$ and $\sigma(G) = n$, it follows that $|E(H)| = n$. Further, since $\Delta(H) \leq 2$, we also have that $|V(H)| = n$. It is also easy to observe that $H$ is a two-factor. Indeed, let $p$ and $q$ be the number of vertices in $H$ that have degree one and two in $H$, respectively. We know that $(p+q) = n$ (note that $H$ has no isolated vertices) and counting the edges, we get $\frac{p+2q}{2} = n$. This gives us $p = 0$.  Thus $H$ is a spanning two-factor in $G$. Observe that the palettes of vertices that are non-adjacent in $H$ are disjoint, that is:

\[ c^{\dag}(u) \cap c^{\dag}(v) = \emptyset \text{ if } (u,v) \notin E(H).\]

Therefore, we have that $G = H$, since any edge $e$ that is not in $H$ is incident on vertices whose palettes with respect to $c$ are disjoint, implying that $c$ cannot be extended to a valid $2$-coloring of $H \cup \{e\}$.
\end{proof}  

\makeatletter{}\section{A \FPT Algorithm for Max Edge $2$-Coloring}
\label{sec:mainFPT}

We begin by describing an exponential kernel for the \name{Exact Edge $2$-Coloring} problem. We will subsequently describe a detailed \FPT algorithm. We first observe that if $G$ has a matching of $k$ edges, then it is already a \YES-instance of the problem. 

\begin{proposition}
\label{Claim:maxmatching}
Let $(G,k)$ be an instance of \name{Exact Edge $2$-Coloring}, and let $m$ denote the number of edges in $G$. If $m < k$, then $G$ is a \NO-instance. If the size of the maximum matching in $G$ is at least $(k-1)$ and $m \geq k$, then $G$ is a \YES-instance. 
\end{proposition}

\begin{proof}
The first part of the claim is trivial: if $m < k$, then it is not possible to color the edges of $G$ with $k$ colors, since every edge is assigned exactly one color. 

On the other hand, let $m \geq k$, and let $M = \{e_1, \ldots, e_r\}$ be the set of edges in a maximum matching of $G$. Consider the following coloring for $G$:

\begin{equation*}
c(e) = \left\{
\begin{array}{cl}
i & \text{if } e = e_i\\
0 & \text{otherwise.}
\end{array} \right.
\end{equation*}

It is easily checked that the coloring $c$ uses $(r+1)$ colors  (note that there is at least one edge that is colored $0$ by $c$) and is $2$-valid. The claim follows. 
\end{proof}

Let $(G=(V,E),k)$ be an instance of \name{Exact Edge $2$-Coloring}. The first step towards an exponential kernel is to identify a matching of maximum size, say $M$, and return a trivial \YES-instance if $|M| \geq k-1$. If this is not the case, let $S \subseteq V$ be the set of both endpoints of every edge in $M$. We use $I$ to denote $V \setminus S$. Note that $|S| \leq 2k-4$ and $I$ is an independent set. At this point, we also remark that on instances where the maximum degree is bounded by $d$, we have that the instance size is bounded by $2kd$. This is interesting because the problem is NP-complete for graphs of constant maximum degree, as can be observed by adapting the reduction in~\cite{AP10}.

For $T \subseteq S$, let $I_T \subseteq I$ denote the set of vertices $v$ in $I$ for which $N(v) = T$. Note that $\{ I_T ~|~ T \subseteq S\}$ forms a partition of $I$ into at most $2^{|S|}$ classes. We are now ready to suggest our first reduction rule.

\begin{enumerate}
\item[{\bf (R1)}] 
For $T \subseteq S$, and let $r := \max\{10,|T|+1\}$. If $|I_T| > r$, delete all but $r$ vertices from $I_T$. The reduced instance has the same parameter as the original.
\end{enumerate}

It is easy to see that this reduction rule may be applied in $O(|I|)$ time. We now prove the correctness of this rule. 

\begin{proposition}
\label{Claim:expkernel_correct} Let $(G,k)$ be an instance of \name{Exact Edge $2$-Coloring}, let $S$ be a vertex cover of $G$ and let $T \subseteq S$. Let $(H,k)$ be the instance obtained by applying {\bf (R1)} to $G$ with respect to $T$. The instances $(G,k)$ and $(H,k)$ are equivalent.
\end{proposition}

\begin{proof}
Let $(G,k)$ be a \YES-instance of \name{Exact Edge $2$-Coloring}, and let $c$ be a $2$-valid edge coloring of $G$ that uses $k$ colors. Notice that since $H$ is a subgraph of $G$, we may restrict $c$ to $H$ to obtain a $2$-valid coloring of $H$. However, it is not clear if such a restriction would use all the $k$ colors. It turns out that this can be ensured by a carefully planned restriction.\footnote{In what follows, we use the same notation for vertex subsets that are common to $G$ and $H$, while ensuring that the graph in question is clear from the context.}

Consider $T \subseteq S$. If $|I_T| \leq  \max\{10,|T|+1\}$, then there is nothing to prove. Therefore, suppose $|I_T| > \max\{10,|T|+1\}$, and consider the palettes of vertices of $I_T$ with respect to $c$, namely, ${\mathcal P}_T := \{ c^\dag(v) ~|~ v \in I_T \}$. Note that, for $v \in I_T$, if $c^\dag(v) = \{p,q\}$, then for every $u \in T$, $c^\dag(u) \cap \{p,q\} \neq \emptyset$.  We now have a case analysis on ${\mathcal P}_T$.

\paragraph{Case 1.} ${\mathcal P}_T$ contains at least one pair of mutually disjoint sets. Note that all sets in ${\mathcal P}_T$ have size at most two. We analyze the situation assuming that both sets have two elements, the same argument (and therefore, the same conclusion) holds for the situations when one or both sets have size one. 

Let $\{p,q\} \in {\mathcal P}_T$ and $\{p',q'\} \in {\mathcal P}_T$ such that $\{p,q\} \cap \{p',q'\} = \emptyset$. Then for every $u \in T$, we have: $$c^\dag(u) \cap \{p,q\} \neq \emptyset \text{ and } c^\dag(u) \cap \{p',q'\} \neq \emptyset.$$ 

We claim that in this situation, for any $v \in I_T$, $c^{\dag}(v) \subseteq \{p,p',q,q'\}$. Indeed, suppose not, and let $i \in c^{\dag}(v)$ and $i \notin \{p,p',q,q'\}$. But this would imply that $i \in c^\dag(u)$ for some $u \in T$. However, $c^\dag(u)$ already contains one of $p$ or $q$ {\em and} one of $p'$ or $q'$, and $|c^\dag(u)|$ is at most two, implying that there is no room for the additional color $i$ --- a contradiction. 

Therefore, ${\mathcal P}_T \subset 2^Q$ where $Q := \{p,p',q,q'\}$. Noting that ${\mathcal P}_T$ contains sets of size at most two, we have that $|{\mathcal P}_T| \leq 10$.

\paragraph{Case 2.} ${\mathcal P}_T$ contains no mutually disjoint sets, in other words, ${\mathcal P}_T$ is a pairwise intersecting family. Notice that ${\mathcal P}_T$ is a family of sets of size at most two over the universe: ${\mathcal C}_T := \bigcup_{u \in T}c^\dag(u)$. 

We first claim that $|{\mathcal C}_T| \leq |T|+2$. To see this, consider the palettes of vertices of $T$ with respect to $c$, namely, ${\mathcal Q}_T := \{ c^\dag(u) ~|~ u \in T \}$. Suppose the maximum number of mutually disjoint sets in ${\mathcal Q}_T$ is two. Then $|{\mathcal C}_T| \leq 4 + (|T| - 2) = |T|+2$. On the other hand, note that ${\mathcal Q}_T$ does not have more than two mutually disjoint sets --- indeed, having three mutually disjoint sets in ${\mathcal Q}_T$ would mean that $c$ cannot be extended to a $2$-valid coloring of any vertex in $I_T$. Therefore, we have that $|{\mathcal C}_T| \leq |T|+2$.

Now we return to ${\mathcal P}_T$, which is an intersecting family of sets of size at most two over an universe of size at most $|T|+2$. It follows that $|{\mathcal P}_T| \leq \max\{|T|+1,3\}$.

\paragraph{}Overall, we conclude that $|{\mathcal P}_T| \leq \max\{10, |T|+1\}$.

In $I_T$, we form a maximal sub-collection $I_T'$ such that all palettes are distinct. Note that $|I_T'| = |{\mathcal P}_T|$. Let $t := |I_T'|$. Now, let $J_T$ be the vertices in $H$ whose neighborhood in $S$ is $T$. Color the first $t$ vertices of $J_T$ according to $I_T'$ and the remaining vertices arbitrarily. Note that this coloring is always possible since $|J_T| = \max\{10, |T|+1\}$, and the first $t$ vertices are always available. It is routine to check that the proposed coloring on $H$ uses every color used by $c$ in $G$. 

In the reverse direction, let $c$ be a $2$-valid edge coloring of $H$ that uses $k$ colors. Let $v$ be a vertex in $G$ that was not affected by {\bf (R1)}. For all edges incident on $v$, we simply mimic the coloring $c$. Let $v$ be a vertex that was deleted according to {\bf (R1)}. Since $v \in I_T$ for some $T$, there is at least one vertex $w \in I_T$ that is not affected by {\bf (R1)}. Color the edges incident on $v$  according to the coloring of $c$ on the edges incident to $w$ (note that because $w$ and $v$ have exactly the same neighborhood in $G$, this extension is $2$-valid). Note that $H$ appears as a subgraph of $G$, and the coloring on $G$ restricted to this subgraph is identical to the coloring of $H$. It follows that all colors used by $c$ in $H$ are used in the coloring that we have proposed for $G$. This concludes the proof. 

\end{proof}

\begin{lemma}
\name{Exact Edge $2$-Coloring} has a kernel on $O(4^k \cdot (2k-4))$ vertices.
\end{lemma}

\begin{proof}
Notice that once reduced with respect to {\bf (R1)}, for every $T \subseteq S$, there are at most $\max\{10,|T|+1\}$ vertices in $G$. Thus, a conservative upper bound on the number of vertices in a reduced instance would be $(|S| + 2^{|S|}|S|)$, and the lemma follows from the fact that $|S| \leq 2k-4$. 
\end{proof}

We now turn to a \FPT algorithm for \name{Exact Edge $2$-Coloring}. See Algorithm~\ref{func:part1} for a pseudocode-based description of the overall algorithm. Recall that the goal is to compute a $2$-valid edge coloring that uses $k$ colors. We begin by using Proposition~\ref{Claim:maxmatching} to accept instances with a maximum matching on at least $k-1$ edges, and reject instances that have fewer than $k$ edges. Otherwise, let $S$ be the vertex cover obtained by choosing both endpoints of a maximum matching. 

The algorithm begins by guessing a palette assignment $\tau$ to the vertices in $S$. First, some simple sanity checks are implemented. Note that if $c$ is a $2$-valid edge coloring of $G$ that uses $k$ colors, and $S$ is a vertex cover of $G$, then $\bigcup_{v\in S} c^\dag(v) = [k]$ (if not, the missing color cannot be attributed to any edge). Therefore, we ensure that $\bigcup_{v\in S} \tau(v) = [k]$. Also, for an edge in $S$, the palettes assigned to the endpoints clearly cannot be disjoint. Therefore, for $u,v \in S$, if $(u,v) \in E$, we ensure that $\tau(u) \cap \tau(v) \neq \emptyset$.

Let $G$ be a \YES-instance of \name{Exact Edge $2$-Coloring}, and suppose $c$ is a $2$-valid edge coloring of $G$ that uses $k$ colors. Let $X_c \subseteq [k]$ be the set of colors used by $c$ on $S$. More formally, $X_c := \bigcup_{e \in G[S]} c(e)$. The second step of the algorithm involves guessing this subset of colors, that is, we consider all possible subsets of $[k]$ as candidates for being the exact set of colors that are realized by some $2$-valid coloring when restricted to $G[S]$.

Let $X \subseteq [k]$ be the colors that are to be realized in $S$. All the colors in $X$ are initially labelled \textit{unused}. Note that for $u,v \in S$, if $(u,v) \in E$, $p_{uv} := \tau(u) \cap \tau(v)$ either has one or two colors. If the intersection has one color, say $i$, and $i \notin X$, then we reject the guess $X$. On the other hand, if $i \in X$, we assign $i$ to the edge between $u$ and $v$ and update the label for $i$ as \textit{used}. Notice that this is a ``forced'' assignment, since this is the only way to extend $c$ to the edge $uv$ while respecting $\tau$. On the other hand, suppose $p_{uv}$ has two colors. If neither of these colors is in $X$, then we may reject this guess. If it has two colors and only one of them is in $X$, then we assign the color in $X$ to $(u,v)$ and update its label as used. Otherwise, we branch on the two possibilities of $c(u,v)$, which come from $p_{uv}$. Note that the count of colors labelled unused in $|X|$ drops by exactly one in both branches, so this is a two-way branching, where the corresponding search tree has depth bounded by $|X|$. This completes the description of the functionality of {\bf CheckTop} (see also Function~\ref{func:part1}).

\begin{function}[h]
\label{func:part1}
\SetAlgoLined
\DontPrintSemicolon
\KwIn{$G,k,S,\tau,c,X$, an instance of \name{Exact Edge $2$-Coloring}, with a partial edge coloring $c$, a vertex cover $S$, palette assignment $\tau$ to vertices in $S$, and a subset of colors $X$.}
\KwOut{\YES~if $G$ has a $2$-valid edge coloring that uses exactly the colors in $X$ in $G[S]$,  while extending $c$ and respecting $\tau$, \NO~otherwise.}
\uIf{$X = \emptyset$}{\KwRet{\YES}}
	    \For{$(u,v) \in G[S]$ and $(u,v) \in E$}{
	        \uIf{$\tau(u) \cap \tau(v) = \{i\}$}{
	            \uIf{$i \notin X$}{\KwRet{\NO}}
	            $c(u,v) := i$\;
	            $X \leftarrow X \setminus \{i\}$\;
	            Continue\;
	        }
    	}
    	\For{$(u,v) \in G[S]$ and $(u,v) \in E$}{
        	\uIf{$\tau(u) = \tau(v) = \{i,j\}$ and $|\{i,j\} \cap X| = \emptyset$}{\KwRet{\NO}}
	        \ElseIf{$\tau(u) = \tau(v) = \{i,j\}$ and $|\{i,j\} \cap X| = 1$}{
	            $\ell := \{i,j\} \cap X$\;
	            $c(u,v) := \ell$\;
	            $X \leftarrow X \setminus \ell$\;
	            Continue\;
	        }
	        \Else{
	            $c_i(u,v) := i$\;
	            $X_i \leftarrow X \setminus i$\;
	            $a = $ CheckTop($G,k,S,\tau,c_i,X_i$);\;
	            $c_j(u,v) := j$\;
	            $X_j \leftarrow X \setminus j$\;
	            $b = $ CheckTop($G,k,S,\tau,c_j,X_j$);\;
                \KwRet{a {\bf or} b}
	        }
    	}
    	
\If{$X \neq \emptyset$}{\KwRet{\NO}}

\caption{CheckTop($G,k,S,\tau,c,X$)}
\end{function}

Finally, we need to realize the colors in $[k]\setminus X$ on the edges that have one endpoint each in $S$ and $G \setminus S$. To this end, we compute the lists of feasible assignments of colors for each vertex in $G \setminus S$, based on $\tau$. In particular, a pair of colors $\{i,j\}$ belongs to the feasibility list $\ell(u)$ of a vertex $u \in G\setminus S$ if there is a way of coloring the edges incident on $u$ with the colors $i$ and $j$ while respecting the palette $\tau$. In other words, one of the colors $i$ or $j$ appears in $\tau(v)$ for every $v \in N(u)$. If such a list is empty, then we know that no feasible extension of $\tau$ exists. On the other hand, if the list contains a unique set, then we may color the edges incident on $u$ according to the unique possibility. 

Other than the special cases above, we know, for the same reasons as in the proof of Proposition~\ref{Claim:expkernel_correct}, that these lists either have constant size, or have one color in common. When the lists have one color in common, then this color can be removed from $[k]\setminus X$, as such a color will be used by any coloring $c$ that respects $\tau$.

For lists $\ell(u)$ of constant size, as long as at least two elements in the list contain a color from  $[k]\setminus X$, we branch on such elements. Note that the depth of branching is bounded by  $[k]\setminus X$ and the width is bounded by a constant (at most $10$, see Proposition~\ref{Claim:expkernel_correct}). If exactly one element in $\ell(u)$ contains a color from $[k]\setminus X$, then we color $u$ according to that element. If no elements in $\ell(u)$ contain colors from $[k]\setminus X$, then color $u$ according to any element in the list of its feasible assignments. 

Finally, we are left with a situation where some colors from $[k]\setminus X$ still need to be assigned, and the only vertices from $G \setminus S$ that are left are those whose lists contain a common color.  Now this is a question of whether every color that remains in $[k]\setminus X$ can be matched to a vertex from $G \setminus S$ whose feasibility list contains that color. To this end, we construct the bipartite graph $H = ((A \cup B),E)$ as follows. The vertex set $A$ has one vertex for every color in $[k] \setminus X$. The vertex set $B$ has one vertex for every $u \in G \setminus S$ for which the feasibility list of $u$ has a common color. For $i \in A$ and $u \in B$, add the edge $(i,u)$ if $\ell(u)$ has a set which contains $i$. Now we compute a maximum matching $M$ in $H$, and it is easy to see that the remaining colors can be realized if and only if $M$ saturates $A$ (see also the pseudocode for function~{\bf CheckAcross}, Algorithm~\ref{func:part2}).

This brings us to the main result of this section.

\begin{theorem}
There is an algorithm with running time $\OO^*((20k)^k)$ \name{Exact Edge $2$-Coloring}.
\end{theorem}

\begin{proof}
The correctness is accounted for in the description of the algorithm. Guessing the palette assignment requires time $\OO^*((k + {k \choose 2})^k)$ and guessing $X \subseteq [k]$ incurs an expense of $2^k$. We note that the only branching steps happen in lines 22---28 of {\bf CheckTop} and lines 47---52 in {\bf CheckAcross}. The former is a two-way branching with a cost of $2^{|X|}$ and the latter is a $10$-way branching with a cost of $10^{|[k]\setminus X|}$. Overall, therefore, the running time of these branching steps is bounded by $10^k$. Therefore, the overall running time is bounded by $\OO^*((20k)^k)$, as desired.
\end{proof}

\makeatletter{}\begin{function*}[h]
\label{func:part2}
\SetAlgoLined
\DontPrintSemicolon
\KwIn{$G,k,S,\tau,c,X$, an instance of \name{Max Edge $2$-Coloring}, with a partial edge coloring $c$, a vertex cover $S$, palette assignment $\tau$ to vertices in $S$, and a subset of colors $X$.}
\KwOut{\YES~if $G$ has a $2$-valid edge coloring that uses the colors in $X$ in $G \setminus E(S)$,  while extending $c$ and respecting $\tau$, \NO~otherwise.}

\begin{multicols}{2}
    \For{$u \in G\setminus S$}{
    $\ell(u) \leftarrow \{\emptyset\}$\;
    ${\mathcal P}(u) = \bigcup_{v \in N(u)}(\tau(v)) \subseteq [k]$\;
    \For{$Y \in {\mathcal P} \cup {{\mathcal P}(u) \choose 2}$}{
        \For{$v \in N(u)$}{
            \uIf{$\tau(v) \cap Y \neq \emptyset$}{
                $\ell(u) \leftarrow \ell(u) \cup Y$\; 
            }
        }
        }
    }

    \For{$u \in G\setminus S$}{
    \uIf{$\ell(u) = \emptyset$}{\KwRet{\NO}}
    \uIf{$|\ell(u)| = 1$}{$Z \in \ell(u)$ \; 
    \For{$v \in N(u)$}{$c(u,v) = \tau(v) \cap Z$}
    $X \leftarrow X \setminus Z$\;
    Continue;
    }
    
    \uIf{$\bigcap_{Y \in \ell(u)} Y \neq \emptyset$}{
        $i := \bigcap_{Y \in \ell(u)} Y$\;
        $X \leftarrow X \setminus \{i\}$\;
        $g(u) := 1$\;}
    \Else{$g(u) := 0$}
    }
    
    ${\mathcal B} := \{u \in G\setminus S ~|~ g(u) = 0\}$\;
    ${\mathcal G} := \{u \in G\setminus S ~|~ g(u) = 1\}$\;
    \For{$u \in {\mathcal B}$}{
        $f = 0$\;
        \For{$Y \in \ell(u)$}{
            \uIf{$Y \cap X = \emptyset$}{
                $f = 1$\;}
        }
        \uIf{$f = 0$}{
            \For{$Y \in \ell(u)$}{
                Color $u$ according to $Y$.\;
                Let $c'$ be the new coloring.\;
                $X' \leftarrow X \setminus Y$.\;
	            $a_Y = $ CheckAcross($G,k,S,\tau,c',X'$);\;
	         }
	         \KwRet{$\max\{a_Y ~|~ Y \in \ell(u)\}.$}
    } 
    }
    
    \For{$u \in {\mathcal G}$}{
        Construct the bipartite graph $H = ((A \cup B),E)$ as follows. The vertex set $A$ has one vertex for every color in $X$. The vertex set $B$ has one vertex for every $u \in {\mathcal G}$. For $i \in A$ and $v \in B$, add the edge $(i,v)$ if $\ell(u)$ has a set which contains $i$. \;
        Compute a maximum matching $M$ in $H$.\;
       \uIf{$M$ saturates $A$}{\KwRet{\YES}}
       \Else{\KwRet{\NO}}            
     }
\end{multicols}
	    
\caption{CheckAcross($G,k,S,\tau,c,X$)}
\end{function*}

\makeatletter{}\section{Parameterizing Below an Upper Bound: A Linear Kernel}
\label{sec:dualparam}

We now address the question of whether a given graph $G = (V,E)$ admits a $2$-valid edge coloring using at least $(n-k)$ colors, where $n := |V|$. In this section, we show a polynomial kernel with parameter $k$. We note that the NP-hardness of the question is implicit in the NP-hardness of the \name{Max Edge $2$-Coloring Problem} shown in~\cite{AP10}.

The kernel is essentially obtained by studying the structure of a \YES-instances of the problem. We argue that if $G$ is a \YES-instance, $c$ is a $2$-valid edge coloring of $G$ using at least $(n-k)$ colors, and $H$ is a character subgraph of $G$ with respect to $c$, then $|V(H)|$ must be at least $(n-k)$, or in other words, $G \setminus H$ is at most $k$. We then proceed to show that the components which are not cycles in $H$ are also bounded. An easy but crucial observation is that any vertex cannot be adjacent to too many vertices whose palettes are disjoint. On the other hand, we are able to bound the number of vertices in $H$ whose palettes are not disjoint. This leads to a bound on the maximum degree of $G$ in terms of $k$. Finally, we show a reduction rule that applies to ``adjacent degree two vertices'', and this finally rounds off the analysis of the kernel size. We now formally describe the sequence of claims leading up to the kernel.

We begin by analyzing the structure of \YES-instances of the problem. Let $G = (V,E)$ be a graph that admits  a $2$-valid edge coloring using at least $(n-k)$ colors. Let $c$ be such a coloring, and let $H$ be a character subgraph with respect to $c$. Since $\Delta(H) \leq 2$, the components of $H$ comprise of paths and cycles. Let $C_1, \ldots, C_r$ denote the components of $H$ that are cycles and let $P_1, \ldots, P_s$ denote the components that are paths. Let the sizes of these components be $c_1, \ldots, c_r, p_1, \ldots, p_s$, respectively. We first claim that $s \leq k$.

\begin{proposition}
\label{claim:few_paths}
Let $c$ be a $2$-valid edge coloring of $G$ using at least $(n-k)$ colors, and let $H$ be a character subgraph with respect to $c$. If $H$ consists of $s$ paths of lengths $p_1, \ldots, p_s$ and $r$ cycles of lengths $c_1, \ldots, c_r$, then $s \leq k$.
\end{proposition}

\begin{proof}

First, note that the number of edges in $H$ is at least $(n-k)$. With this as a starting point, we obtain:

\begin{IEEEeqnarray*}{rCl}
\sum_{i=1}^r c_i + \sum_{i=1}^s (p_i-1) & \geq & (n-k) \\
\underbrace{\left(\sum_{i=1}^r c_i + \sum_{i=1}^s p_i\right)}_{= |V(H)| \leq |V(G)|} - s & \geq & (n-k) \\
(n-s) & \geq & (n-k)\\
s & \leq & k
\end{IEEEeqnarray*}
\end{proof}

Next, we show that there are at most $k$ vertices in $G$ that are not in $H$.

\begin{proposition}
\label{claim:H_is_big}
Let $c$ be a $2$-valid edge coloring of $G$ using at least $(n-k)$ colors, and let $H$ be a character subgraph with respect to $c$. Then, $|G \setminus H| \leq k$.
\end{proposition}

\begin{proof}
Suppose, for the sake of contradiction, that $|G \setminus H| > k$. This in turn implies that $|H| < n-k$. Recall, however, that  $\Delta(H) \leq 2$, and therefore $|E(H)| \leq \frac{2|H|}{2} = |H| < n-k.$ However, since $H$ is character subgraph of $G$ with respect to a coloring that uses at least $(n-k)$ colors, we have that $|E(H)| \geq n-k$. Therefore, the above amounts to a contradiction.\end{proof}

Let ${\mathcal P}$ denote the set of endpoints of the paths $P_1, \ldots, P_s$. Notice that $|{\mathcal P}| \leq 2k$. Let ${\mathcal T}$ denote the remaining vertices in $H$, that is, ${\mathcal T} := H \setminus {\mathcal P}$. We now claim that the maximum degree of $G$ is bounded:

\begin{proposition}
\label{claim:deg5inT}
For a graph $G$ that admits a $2$-valid edge coloring using at least $(n-k)$ colors, its character subgraph is such that, any vertex $u$ in $G$ is adjacent to at most six vertices in ${\mathcal T}$.
\end{proposition}

\begin{proof}
Let $c$ be a $2$-valid edge coloring of $G$ using at least $(n-k)$ colors, and let $H$, ${\mathcal P}$ and ${\mathcal T}$ be defined as above.

Suppose, for the sake of contradiction, that there is a vertex $u \in G$ that has more than six neighbors in ${\mathcal T}$. Since $\Delta(H[{\mathcal T}]) \leq 2$, in any subset of seven vertices of ${\mathcal T}$, there is at least one triplet of vertices, say $x,y,$ and $z$ that are mutually non-adjacent in $H$. By definition of $H$ and ${\mathcal T}$, we know that the palettes of $x$, $y$ and $z$ with respect to $c$ have two colors each and are mutually disjoint:

$$c^{\dag}(x) \cap c^{\dag}(y) = \emptyset; c^{\dag}(x) \cap c^{\dag}(z) = \emptyset; \text{ and } c^{\dag}(y) \cap c^{\dag}(z) = \emptyset.$$

It follows that $|c^{\dag}(x)| = |c^{\dag}(y)| = |c^{\dag}(z)| = 2$. Since $u$ is adjacent to $x,y$ and $z$, we conclude that there is no way to extend $c$ to a $2$-valid coloring of the edges $(u,x)$, $(u,y)$ and $(u,z)$. Therefore, we contradict our assumption that $c$ is a $2$-valid edge coloring of $G$ using at least $(n-k)$ colors, and conclude that all vertices in $G$ have at most six neighbours in ${\mathcal T}$.\end{proof}

The following corollary is implied by the fact that there are at most $3k$ vertices in the graph other than ${\mathcal T}$. 

\begin{corollary}
\label{claim:Delta_is_small}
Let $G$ be a graph that admits a $2$-valid edge coloring using at least $(n-k)$ colors. Then $\Delta(G) \leq 3k+6$.
\end{corollary}

We now state the reduction rules that define the kernelization.

\begin{enumerate}
\item[{\bf (R1)}] If $\Delta(G) > 3k+6$, then return a trivial \NO-instance.
\item[{\bf (R2)}] Let $u$ and $v$ be adjacent vertices with $d(u) = d(v) = 2$, and let $v'$ be the other neighbor of $v$. Delete $v$ and add the edge $(u,v')$. Let the graph obtained thus be denoted by $H$. Then the reduced instance is $(H,n^*-k)$, where $n^* = |V(H)| = (n-1)$. Notice that the parameter does not change.
\end{enumerate}

\makeatletter{}
\begin{figure}[h!]
\begin{center}
\begin{tikzpicture}[thick]
\tikzstyle{every node}=[circle, fill=SlateBlue,
                        inner sep=0pt, minimum width=4pt];

\draw \foreach \x in {140,170,200,230} {
        ([shift={(9,0.3)}]\x:1.2) node (A-\x) {} 
 };
\tikzstyle{every node}=[circle, fill=SlateBlue,
                        inner sep=0pt, minimum width=6pt,shift={(9,0.3)}];

\node[minimum width=10pt,fill=YellowGreen] (e) at (0,0) {{\color{white}$v'$}};

\tikzstyle{every node}=[circle, fill=SlateBlue,
                        inner sep=0pt, minimum width=4pt,shift={(0,0)}];

\draw \foreach \x in {140,170,200,230}{
(A-\x) -- (e)
};

\tikzstyle{every node}=[circle, fill=YellowGreen,
                        inner sep=0pt, minimum width=5pt,shift={(9,0.3)}]

\node[minimum width=10pt,fill=YellowGreen] (f1) at (2,0) {{\color{white}$v$}};
\draw (e) -- (f1);

\node[minimum width=10pt,fill=YellowGreen] (f2) at (4,0) {{\color{white}$u$}};
\draw (f1) -- (f2);

\node[minimum width=10pt,fill=SlateBlue] (g) at (6,0) {};
\draw (f2) -- (g);

\tikzstyle{every node}=[circle, fill=SlateBlue,
                        inner sep=0pt, minimum width=4pt,shift={(6,0)}]
                        
\draw \foreach \x in {40,10,340,310} {
        ([shift={(9,0.3)}]\x:1.2) node (B-\x) {} 
 };

\draw \foreach \x in {40,10,340,310}{
(B-\x) -- (g)
};

\tikzstyle{every node}=[circle, fill=SlateBlue,
                        inner sep=0pt, minimum width=6pt]
\end{tikzpicture}
\end{center}

\[ \Downarrow \]
\begin{center}
\begin{tikzpicture}[thick]
\tikzstyle{every node}=[circle, fill=SlateBlue,
                        inner sep=0pt, minimum width=4pt]

\draw \foreach \x in {140,170,200,230} {
        ([shift={(9,0.3)}]\x:1.2) node (A-\x) {} 
 };
\tikzstyle{every node}=[circle, fill=SlateBlue,
                        inner sep=0pt, minimum width=6pt,shift={(9,0.3)}]

\node[minimum width=10pt,fill=YellowGreen] (e) at (0,0) {{\color{white}$v'$}};

\tikzstyle{every node}=[circle, fill=SlateBlue,
                        inner sep=0pt, minimum width=4pt,shift={(0,0)}]

\draw \foreach \x in {140,170,200,230}{
(A-\x) -- (e)
};

\tikzstyle{every node}=[circle, fill=SlateBlue,
                        inner sep=0pt, minimum width=4pt,shift={(9,0.3)}]

\node[minimum width=10pt,fill=YellowGreen] (f2) at (4,0) {{\color{white}$u$}};
\draw (e) -- (f2);

\node[minimum width=10pt,fill=SlateBlue] (g) at (6,0) {};
\draw (f2) -- (g);

\tikzstyle{every node}=[circle, fill=SlateBlue,
                        inner sep=0pt, minimum width=4pt,shift={(6,0)}]
                        
\draw \foreach \x in {40,10,340,310} {
        ([shift={(9,0.3)}]\x:1.2) node (B-\x) {} 
 };

\draw \foreach \x in {40,10,340,310}{
(B-\x) -- (g)
};

\tikzstyle{every node}=[circle, fill=YellowGreen,
                        inner sep=0pt, minimum width=6pt]
\end{tikzpicture}
\caption{A reduction rule for adjacent degree two vertices.}
\label{fig:degree_two}
\end{center}
\end{figure}
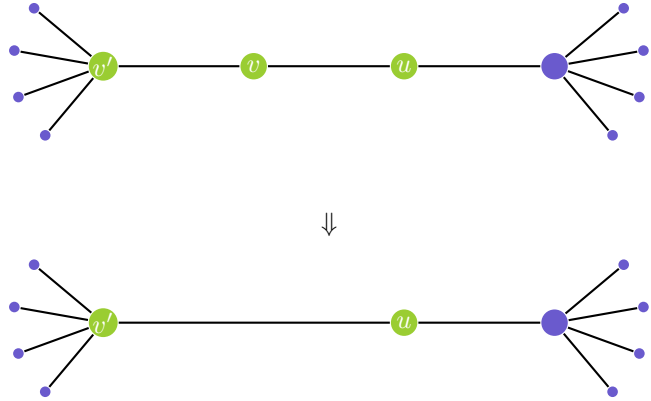

It is easy to see that both the reduction rules above can be executed in linear time. The correctness of {\bf (R1)} follows from Corollary~\ref{claim:Delta_is_small}. We now show the correctness of the second reduction rule.

\begin{proposition}
\label{claim:r2_is_correct}
Let $G$ be a graph where vertices $u$ and $v$ are adjacent, and $d(u) = d(v) = 2$. Let $v'$ be the other neighbor of $v$. Let $H$ be the graph obtained from $G$ after an application of reduction rule {\bf (R2)}. The graph $G$ has a $2$-valid edge coloring that uses at least $(n-k)$ colors if and only if the graph $H$ has a $2$-valid edge coloring that uses at least $(n-k-1)$ colors.

\end{proposition}

\begin{proof}

We begin by assuming that $G$ has a $2$-valid edge coloring that uses $(n-k)$ colors. Let $c$ be such a coloring for $G$.  We propose the following edge coloring for $H$: the edge $(u,v')$ is colored with $c(v,v')$, and all other edges are colored according to the color given by $c$ to the corresponding edges in $G$. More precisely,
\begin{equation*}
c^*(e) = \left\{
\begin{array}{cl}
c(v,v') & \text{if } e = (u,v')\\
c(e) & \text{otherwise.}
\end{array} \right.
\end{equation*}

Notice that $c^*$ uses at least $(n-k-1)$ colors, since at most one color was lost in the recoloring, namely the color given to the edge $(u,v)$. It is easily checked that $c^*$ is $2$-valid.

On the other hand, suppose $H$ has a $2$-valid edge coloring, $c^*$, that uses at least $(n-k-1)$ colors. Suppose the set of colors used is $[l]$, where $l\geq n-k-1$. Now consider the following coloring for $G$ based on $c^*$:

\begin{equation*}
c(e) = \left\{
\begin{array}{cl}
c^*(u,v') & \text{if } e = (v',v)\\
l+1 & \text{if } e = (u,v)\\
c^*(e) & \text{otherwise.}
\end{array} \right.
\end{equation*}

In other words, we color the edge $(u,v)$ with a new color, and color $(v,v')$ with the color used on the edge $(u,v')$, and all other edges are colored by mimicking $c^*$. The coloring $c^*$ is clearly $2$-valid: the only vertices that are affected are $v'$, $v$ and $u$, but $v$ and $u$ are vertices of degree two and therefore any coloring is $2$-valid with respect to them, and we observe that $c(F_{v'}) = c^*(F_{v'})$. Notice that $c$ uses at least $(n-k)$ colors, one more than the number of colors used by $c^*$.
\end{proof}

Observe that Proposition~\ref{claim:r2_is_correct} implies the correctness of {\bf (R2)}. We now turn to an analysis of the size of the kernel.

\begin{lemma}
If $(G,n-k)$ is a \YES-instance of \name{$(n-k)$-Edge $2$-Coloring} that is reduced with respect to {\bf (R2)}, then $|V(G)| = O(k)$.
\end{lemma}
\begin{proof}
Since $G$ is a \YES-instance, it admits  a $2$-valid edge coloring $c$ using at least $(n-k)$ colors. Let $H$ be a character subgraph with respect to $c$. Let $C_1, \ldots, C_r$ denote the components of $H$ that are cycles and let $P_1, \ldots, P_s$ denote the components that are paths.

Let ${\mathcal P}$ denote the set of endpoints of the paths $P_1, \ldots, P_s$ and let ${\mathcal T}$ denote the remaining vertices in $H$, that is, ${\mathcal T} := H \setminus {\mathcal P}$. Let $|{\mathcal P}_1|=|G\setminus H|+|{\mathcal P}|$. By Proposition~\ref{claim:deg5inT}, we know that every vertex in $G$, has at most six neighbors in ${\mathcal T}$. Since $|{\mathcal P}_1| \leq 3k$ (this follows from Proposition~\ref{claim:few_paths}), the number of vertices in ${\mathcal T}$ that have neighbors in ${\mathcal P}_1$ is at most $3k \cdot 6 = 18k$. Notice that all other vertices in ${\mathcal T}$ have degree two in $G$. Therefore, we conclude that the number of vertices of $G$ that have degree three or more is at most $3k + 18k = 21k$.

We now have that $|{\mathcal P}| \leq 2k$ and $|G \setminus H| \leq k$, hence it remains to bound the vertices in ${\mathcal T}$. Notice that the vertices of ${\mathcal T}$ have degree two or more in $G$. Among them, the vertices that have degree three or more in $G$ are bounded by $21k$. The vertices left are the vertices in ${\mathcal T}$ that have degree two in $G$. Since the graph is reduced with respect to {\bf(R2)}, the neighbors of these vertices have either degree one or degree three or more. Note that the number of degree one vertices is at most $|{\mathcal P}| + |G \setminus H| \leq 3k$. Hence the number of degree two vertices in ${\mathcal T}$ is at most $21k \cdot 5 + 3k = 108k$. Thus the total number of vertices in ${\mathcal T}$ is also $\OO(k)$. This concludes our argument.
\end{proof} 

\makeatletter{}\section{$C_4$-free graphs}
\label{sec:kernels}

In this section, we discuss the problem restricted to the class of $C_4$-free graphs. We first show that the \name{Max Edge $2$-Coloring} problem is NP-hard on $C_4$-free graphs, and then describe a simple argument for a polynomial kernel on this graph class.

\makeatletter{}
\subsection{NP-Hardness}

In this section, we show the NP-hardness of \name{Max Edge $2$-Coloring} on graphs that have no cycles of length four. We will first need to introduce some definitions. For a function $f: V(G) \rightarrow \{1,2\}$, we say that $G$ is {\em $f$-valid} if every vertex $v$ in the graph is incident to edges colored with at most $f(v)$ distinct colors.  We will consider the following variant of the original problem: 

\problembox{Max Edge $\{1,2\}$-Coloring}{A graph $G$, a function $f: V(G) \rightarrow \{1,2\}$, and an integer $k$}{$k$}{Is $\sigma(G) \geq k$, that is, is there a $f$-valid edge coloring of $G$ with at least $k$ colors?}

It is known from~\cite{AP10} that there is a polynomial time reduction from \name{Max Edge $\{1,2\}$-Coloring} to \name{Max Edge $2$-Coloring}. If $G$ is an instance of \name{Max Edge $\{1,2\}$-Coloring}, the reduction proposed in~\cite{AP10} modifies $G$ by adding a pendant neighbor to every vertex $v$ for which $f(v) = 1$. It is shown that this graph has a $2$-valid coloring using at least $k$ colors if, and only if, $G$ has a $f$-valid coloring using at least $k$ colors.

\begin{proposition}[Theorem 5,~\cite{AP10}]
\label{prop:12-to-2-hard}
There is a polynomial time (and parameter preserving) reduction from \name{Max Edge $\{1,2\}$-Coloring} to \name{Max Edge $2$-Coloring}.
\end{proposition}

We will rely on this reduction to obtain the our NP-hardness results, as it simplifies our presentation greatly when we reduce to \name{Max Edge $\{1,2\}$-Coloring} instead of the original question.

We reduce from the \name{Multi-Colored Independent Set} problem, which is known to be NP-complete by a reduction from the classical NP-complete problem of finding the largest independent set. The ``multi-colored'' variant is defined as follows:

\problembox{Multi-Colored Independent Set}{A graph $G=(V,E)$, and a partition of $V$ into $k$ parts, $V= V_{1} \biguplus  \cdots \biguplus V_{k}$}{$k$}{Does there exist a subset $S\subseteq V$ such that $G[S]$ forms an independent set and $|S\cap V_{i}|=1$, for all $1 \le i \le k$?}

We first describe the construction. Let $G = (V,E)$, and $V = V_1 \biguplus \cdots \biguplus V_k$ be an instance of \name{Multi-Colored Independent Set}. Let $(H,f,l)$ denote the reduced instance; we will first describe $H$, then determine $f$ and finally fix $l$. 

To begin with, the vertex set of $H$ contains one vertex for every vertex in $G$. If $v \in G$, then we abuse notation and denote the corresponding vertex in $H$ by $v$ as well. We also add the vertices $\{g_1,\ldots,g_k\}$ and make $g_i$ a common neighbor of all vertices in $V_i$. In other words, we introduce the edges $\{(g_i,v) ~|~ v \in V_i\}$ for all $1 \leq i \leq k$. Then introduce a vertex $g$ which is adjacent to every $g_i$. For every $e \in E(G)$, where $e = (u,v)$, introduce five new vertices $\{e_u, e_u', e_3, e_v', e_v\}$, and add the edges: $(u,e_u), (g,e_3), (v,e_v), (e_u,e_{u'}), (e_{u'},e_3), (e_3,e_{v'})$ and $(e_{v'},e_v)$. This completes the description of $H$ (see also Figure~\ref{fig:c4-free-reduction}).

\makeatletter{}\begin{figure}
\label{fig:c4-free-reduction}
\centering
\begin{tikzpicture}
[outline/.style={color=SlateBlue,thin},
outline happy/.style={color=YellowGreen,thin},
general/.style={color=black,thin},]

\foreach \x in {0,2,4,8,10,12}
\draw[general] (\x,0) rectangle (\x+1,4);

\foreach \x in {0,2,4,8,10,12}{
    \foreach \y in {0.5,1,...,3.5}{
        \draw[outline] node (\x-\y) [shape=circle,fill=white,draw] at (\x+.5,\y) {};
     }
}

\draw[outline happy] node (1) [shape=circle,fill=white,draw] at (0.5,5) {$g_1$};
\draw[general] node [shape=circle,draw=none] at (2.5,5) {$\ldots$};
\draw[outline happy] node (2) [shape=circle,draw] at (4.5,5) {$g_i$};
\draw[general] node [shape=circle,draw=none] at (6.5,5) {};
\draw[outline happy] node (3) [shape=circle,draw] at (8.5,5) {$g_j$};
\draw[general] node [shape=circle,draw=none] at (10.5,5) {$\ldots$};
\draw[outline happy] node (4) [shape=circle,draw] at (12.5,5) {$g_k$};

\foreach \y in {1,2,3,4}{
    \draw[general] (\y) -- +(-45:.6cm);
    \draw[general] (\y) -- +(-75:.6cm);
    \draw[general] (\y) -- +(-60:.6cm);
    \draw[general] (\y) -- +(-90:.6cm);
    \draw[general] (\y) -- +(-105:.6cm);
    \draw[general] (\y) -- +(-120:.6cm);
    \draw[general] (\y) -- +(-135:.6cm);
}

\node[outline,shape=circle,draw] (global)  at (6.5,7) {$g$};

\foreach \x in {-4,-2,0,2}
    \draw[general] (6.5+\x,-3) -- (6.5+\x+2,-3) {};
\foreach \x in {-4,-2,2,4}
    \node[outline,shape=circle,draw,fill=white] (n\x) at (6.5+\x,-3) {};

\node[outline happy,shape=circle,draw,fill=white] (n0) at (6.5,-3) {};
    
\draw[general] (n-4) -- (4-1);
\draw[general] (n4) -- (8-2);

\foreach \y in {1,2,3,4}
    \draw[general] (\y) -- (global);

\draw[general] (global) -- (n0);

\end{tikzpicture}
\caption{A description of the reduction from \name{Multi-Colored Independent Set} to \name{Max Edge $\{1,2\}$-Coloring}.}

\end{figure}
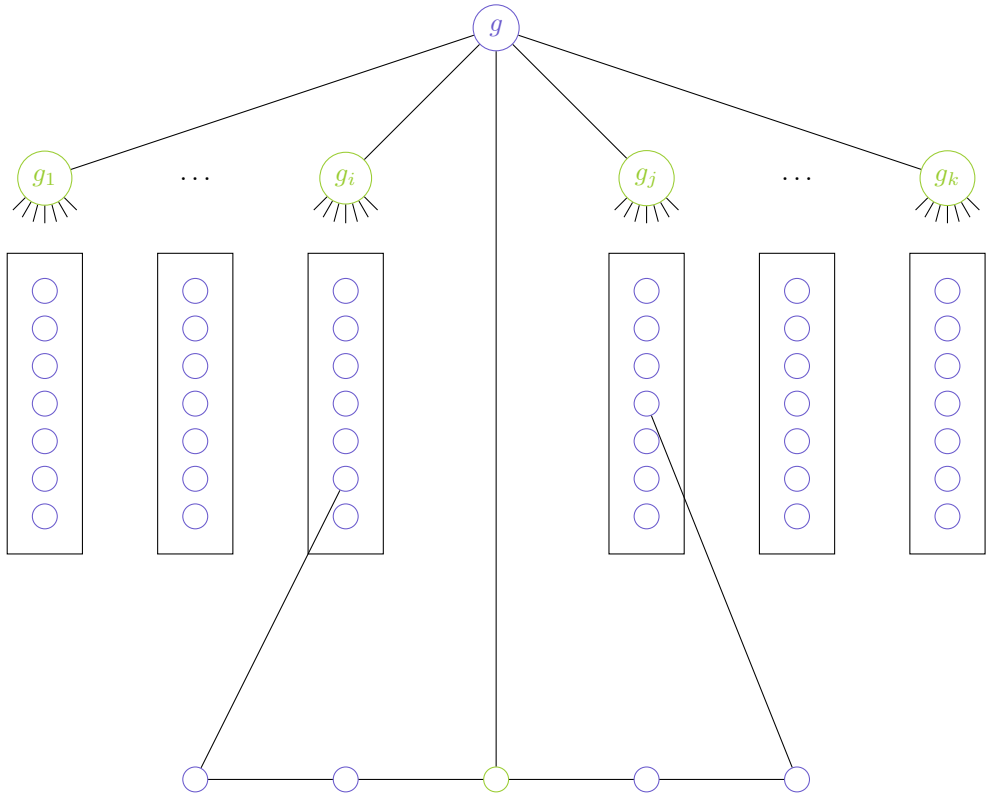

We define $f$ as follows. For any $e \in G$, $f(e_3) = 2$. Also, $f(g_i) = 2$ for all $1 \leq i \leq k$. For any other vertex $v \in H$, we have $f(v) = 1$. Let $l := k+1$. We are now ready to show the equivalence of $(G,V_1 \biguplus \cdots \biguplus V_k,k)$ and $(H,f,l)$.

\begin{lemma}
\label{lem:redn-c4-free}
The graph $G$ with vertex partitions $V_1, \ldots, V_k$ has a multi-colored independent set of size $k$ if and only if $H$ has a $f$-valid coloring that uses at least $k+1$ colors.
\end{lemma}

\begin{proof}
In the forward direction, let $S = \{u_1,\ldots,u_k\}$ be a multi-colored independent set. Without loss of generality, let $u_i \in V_i$. Consider the following coloring function (the notation ``$\star$'' is used to denote any vertex): 

\begin{equation*}
c(t) = \left\{
\begin{array}{rl}
i & \text{if } t = (\star,u_i), \text{ or } t \in \{(u_i,e_{u_i}),(e_{u_i},e_{u_i}'),(e_{u_i}',e_3) ~|~ e \in E(G) \}\\
0 & \text{otherwise}.
\end{array} \right.
\end{equation*}

It is clear that $c$ uses $(k+1)$ colors. We now argue that $c$ is $f$-valid. Recall that we use $F_v$ to denote the set of edges incident on a vertex $v$, and $c(F)$ to denote the set of colors assigned to the edge set $F$. First, observe that $c(F_g) = \{0\}$. Further, for every $u\in G$, we have either that $c(F_u) = \{i\}$ (if $u \in V_i \cap S$), or $c(F_u) = \{0\}$. Now consider the vertex $g_i$ for any $1 \leq i \leq k$. Clearly, $N(g_i) \cap \{u_1, \ldots, u_k\} = \{u_i\}$. Therefore, $c(F_{g_i}) = \{0,i\}$. Recall that is $f$-valid since $f(g_i) = 2$ for all $1 \leq i \leq k$.

Finally, for $e \in G$, and $e=(u,v)$, consider the vertices $e_u, e_{u}', e_3, e_{v}',$ and $e_v$. If $u \notin S$ and $v \notin S$, then it is easy to verify that $$c(F_{e_u}) = c(F_{e_{u}'}) = c(F_{e_3}) = c(F_{e_{v}'}) = c(F_{e_v}) = \{0\}.$$ The other case is that one of $u$ or $v$ belong to $S$ (indeed, since $G[S]$ is an independent set, $S$ does not contain {\em both} $u$ and $v$). Without loss of generality, let $u \in S$ and suppose $u \in V_i$. Then note that $$c(F_{e_u}) = c(F_{e_{u}'}) = \{i\}, c(F_{e_3}) = \{0,i\}, \text{ and } c(F_{e_{v}'}) = c(F_{e_v}) = \{0\}.$$ 

Recall that $f(e_3) = 2$, and therefore, the above is a $f$-valid coloring. This completes the argument in the forward direction.

We now turn to the reverse direction. Let $c$ be a coloring of $H$ that uses $(k+1)$ colors, $\{0,1,\ldots,k\}$. Recall that $f(g) = 1$, and let $0$ denote the unique color in $c(F_g)$. We claim that:

\[ {\mathcal C} := \bigcup_{1 \leq i \leq k} c(F_{g_i}) = [k] \cup \{0\}. \]

The argument is by contradiction. Suppose the above does not hold, and let $j$ be the smallest color that is missing from ${\mathcal C}$. Let $e$ be the edge for which $c(e) = j$ (there must be such an edge, since $c$ uses all the $(k+1)$ colors). The edge cannot be incident to any of the $g_i$'s by assumption. Therefore, it must be one of the edges: $(e_u,e_u'),(e_u',e_3),(e_3,e_v')$ or $(e_v',e_v)$, for some $e = (u,v) \in E(G)$. Now, notice that if $c(e_u',e_3) = j$, then $c(e_u,e_u') = j$, since $f(e_u') = 1$. Similarly, since $f(e_u) = 1$, we have that $c(e_u,u) = j$. Finally, since $c(u) = 1$, $c(u,g_i) = j$ for some $i \in [k]$. This contradicts our assumption that there is no $g_i$ for which $j \in c(F_{g_i})$. An analogous argument can be made if $c(e_v',e_3) = j$, $c(e_u',e_u) = j$, or $c(e_v',e_v) = j$.

Therefore, we have that $c(F_{g_i}) = \{a_i,0\}$, where $a_i \in [k]$. Notice that $a_i \neq a_j$ for any $i \neq j$ (since $\{a_1, \ldots, a_k\} = [k]$). For $1 \leq i \leq k$, let $u_i \in V_i$ be a vertex for which $c(g_i,u_i) = a_i$. Notice that there is at least one such vertex, since all the neighbors of $g_i$ are in $V_i \cup \{g\}$, and one of the edges incident on $g_i$ is colored with $a_i$ (note that this edge cannot be $(g,g_i)$, which we already know has color $0$ according to $c$).  Consider $S$ defined as $\{u_1, \ldots, u_k\}$. By definition, we have that $|S \cap V_i| = 1$ for all $1\leq i \leq k$. We now argue that $G[S]$ is an independent set.

Suppose not, and let $e := (u_p,u_q) \in E$, for $1 \leq p \neq q \leq j$. By the definition of $S$, we have that $c(g_p,u_p) = a_p$ and $c(g_q,u_q) = a_q$, where $a_p \neq a_q \neq 0$. Since $f(u_p) = f(u_q) = 1$, observe that $c(u_p,e_{u_p}) = a_p$ and $c(u_q,e_{u_q}) = a_q$. Further, since$f(e_{u_p}) = f(e_{u_q}) = 1$, again, $c(e_{u_p},e_{u_p}') = a_p$ and $c(e_{u_q},e_{u_q}') = a_q$. Finally, since $c(e_{u_p}') = c(e_{u_q}') = 1$, we obtain:

$$c(e_{u_p}',e_3) = a_p \text{ and } c(e_{u_q}',e_3) = a_q,$$

Recall that $(e_3,g) \in E(H)$ and therefore, $0 \in c(F_{e_3})$. Putting this together with the above, we see that $c(F_{e_3}) = \{0,a_p,a_q\}$, contradicting our assumption that $c$ was $f$-valid. 
\end{proof}

It follows from the construction that $H$ as defined above has no cycles of length four. Observe that adding pendant vertices does not add any cycles to the graph. Therefore, we have the following result by Proposition~\ref{prop:12-to-2-hard}, Lemma~\ref{lem:redn-c4-free}.

\begin{theorem}
\label{thm:c4_free_hardness}
The \name{Max Edge $2$-Coloring} problem is NP-complete on graphs that have no cycles of length four.
\end{theorem}

As an aside, we note that the reduction in this section also shows the W[1]-hardness of the Max Edge $\{1,2\}$-Coloring variant. The fact that \name{Max Edge $2$-Coloring} is FPT is not surprising, since the reduction from \name{Max Edge $\{1,2\}$-Coloring} to \name{Max Edge $2$-Coloring} is not parameter-preserving.

\begin{corollary}
The \name{Max Edge $\{1,2\}$-Coloring} problem is W[1]-hard on graphs that have no cycles of length four.
\end{corollary}

\subsection{A Polynomial Kernel for $C_4$-free graphs}

We now turn to an argument for a polynomial kernel on the class of $C_4$-free graphs. As with the exponential kernel on general graphs, the first step is to identify a matching of maximum size, say $M$, and return a trivial \YES-instance if $|M| \geq k-1$. If this is not the case, let $S \subseteq V$ be the set of both endpoints of every edge in $M$. We use $I$ to denote $V \setminus S$. 

Consider any vertex $v \in S$ and let $X_v$ denote the neighbors of $v$ in $I$, that is, $X_v := N(v) \cap I$. Observe that for any two vertices $u$ and $v$ in $S$, $X_u$ and $X_v$ can share at most one vertex. Indeed, if $p,q \in X_u \cap X_v$, then $\{u,p,v,q\}$ form a four-cycle, contradicting the assumption that $G$ is  $C_4$-free. We say that a vertex $u \in X_v$ is a {\em shared neighbor} of $v$ if there exists a $w \in S$, $w \neq v$, such that $u \in X_w$. Notice that $v$ has at most $(2k-1)$ shared neighbors. 

Let $Y_v \subseteq X_v$ denote the set of shared neighbors of $v$, and let $Z_v := X_v \setminus Y_v$ denote the rest of $X_v$. Notice that every vertex $u \in Z_v$ is a degree one vertex in the graph $G$ (since it has no other neighbors in either $S \setminus \{v\}$ or $I \setminus \{u\}$ by definition). It is easy to see that if $|Z_v| > 2$, then we may safely delete all but two vertices from $Z_v$ and obtain an equivalent instance. 

Therefore, in a reduced instance, $|X_v| \leq 2k+1$ for all $v \in S$, and since $|I| \leq \sum_{v \in S} |X_v|$, we have that $|I| \leq (2k+1) \cdot 2k$, and $|V| \leq |I| + |S| \leq 2k(2k+2) = O(k^2)$. We thus have the following.

\begin{lemma}
The \name{Max Edge $2$-Coloring} problem has a quadratic vertex kernel when restricted to $C_4$-free graphs. 
\end{lemma}

\makeatletter{}\section{Concluding Remarks and Future Work} The most natural unresolved question is to settle the kernelization complexity of the maximum edge $2$-coloring problem when parameterized by the solution size. The exponential kernel described in this work implies a polynomial kernel when the input is restricted to graphs where the maximum degree is a constant, and also if the input is restricted to graphs without cycles of length four. These observations are interesting because the problem continues to be \NPC~for both of these graph classes. The NP-hardness for graphs of bounded degree can be obtained by easy modifications to the reduction proposed in~\cite{AP10}, and we have described the NP-hardness on graphs without cycles of length four. Given these results, the question of whether the problem admits a polynomial kernel on general graphs is an interesting open problem. 

Improved FPT algorithms for both the standard and the dual parameter, specifically with running time $O(c^k)$ for some constant $c$, will be of interest as well. It is also natural to pursue the above-guarantee version of the question, with the size of the maximum matching used as the guarantee. In particular, if $\gamma$ is the size of a maximum matching in a graph $G$, we would like to study the question of checking if $G$ can be colored with at least $(\gamma + k)$ colors, parameterized by $k$. 

For the more general question of \name{Maximum Edge $q$-Coloring}, note that since the problem is NP-complete for fixed values of $q$, the question is para-NP-complete when parameterized by $q$ alone. Generalizing some of the results that hold for $q=2$ is also an interesting direction for future work.

\bibliographystyle{plain}
\bibliography{maxedgecoloring}

\end{document}